\documentclass[12pt]{article}

\usepackage{amssymb,amsmath}
\usepackage{color}
\usepackage{graphicx}
\usepackage{tikz}
\usepackage{epstopdf}

\usepackage[normalem]{ulem}
\usepackage{amsthm}
\usepackage{algorithm}

\usepackage[noend]{algpseudocode}

\def\eb{{\bf e}}

\def\B{{\cal B}}

\def\E{{\cal E}}

\def\N{{\cal N}}

\def\V{{\cal V}}

\def\Fbb{{\mathbb F}}

\def\R{{\mathbb R}}

\def\al{\alpha}
\def\d{\delta}
\def\D{\Delta}
\def\e{\epsilon}
\def\g{\gamma}
\def\GA{\Gamma}

\def\r{\rho}
\def\s{\sigma}
\def\SI{\Sigma}
\def\t{\tau}
\def\th{\theta}

\def\ap{\rightarrow}

\def\seq{\subseteq}

\def\bi{\{0,1\}}

\def\fa{\; \forall}

\def\st{\mbox{ s.t. }}

\def\nm{\Vert}

\renewcommand{\and}{\mbox{$\wedge$}}

\newcommand{\bc}{\begin{center}}
\newcommand{\ec}{\end{center}}
\newcommand{\be}{\begin{equation}}
\newcommand{\ee}{\end{equation}}
\newcommand{\bd}{\begin{displaymath}}
\newcommand{\ed}{\end{displaymath}}
\newcommand{\ba}{\begin{array}}
\newcommand{\ea}{\end{array}}
\newcommand{\ben}{\begin{enumerate}}
\newcommand{\een}{\end{enumerate}}
\newcommand{\bit}{\begin{itemize}}
\newcommand{\eit}{\end{itemize}}
\newcommand{\beq}{\begin{eqnarray}}
\newcommand{\eeq}{\end{eqnarray}}
\newcommand{\btab}{\begin{tabular}}
\newcommand{\etab}{\end{tabular}}
\newcommand{\bfig}{\begin{figure}}
\newcommand{\efig}{\end{figure}}
\newcommand{\btp}{\begin{tikzpicture}}
\newcommand{\etp}{\end{tikzpicture}}

\newcommand{\argmin}{\operatornamewithlimits{argmin}}

\newcommand{\nmm}[1]{ \nm #1 \nm }
\newcommand{\nmeu}[1]{ \nm #1 \nm_2 }
\newcommand{\nmeusq}[1]{ \nm #1 \nm_2^2 }

\newcommand{\nmp}[1]{ \nm #1 \nm_p }

\newcommand{\IP}[2]{ \langle #1 , #2 \rangle }

\newcommand{\supp}{\mbox{supp}}

\def\xh{\hat{x}}

\def\nmsl1{\nm_{{\rm SL1}}}

\def\ct{\tilde{c}}
\def\Dcal{{\cal D}}
\def\oneb{\mathbf{1}}
\def\R{{\mathbb R}}
\def\yb{\bar{y}}

{\bf}{\it}
\newtheorem{definition}{Definition}{\bf}{\it}
{\bf}{\rm}
{\bf}{\it}
\newtheorem{theorem}{Theorem}{\bf}{\it}
{\bf}{\it}
{\bf}{\it}
{\bf}{\rm}

\begin{document}

% \IEEEoverridecommandlockouts

\title{
A Fast Noniterative Algorithm for \\
Compressive Sensing Using \\
Binary Measurement Matrices
}

\author{Mahsa~Lotfi and
Mathukumalli~Vidyasagar
% ~\IEEEmembership{Life~Fellow,~IEEE}
\thanks{The authors are with
the Erik Jonsson School of Engineering and Computer Science,
The University of Texas at Dallas, Richardson, TX 75080, USA.
This research was supported by the National Science Foundation, USA under
Award \#ECCS-1306630. 
% and by the Cancer Prevention and Research Institute of Texas (CPRIT)
% under Award RP-140517.
}
}

\maketitle

\begin{abstract}

In this paper we present a new algorithm for compressive sensing that
makes use of \textit{binary} measurement matrices and achieves
\textit{exact} recovery of ultra sparse vectors, in a single pass and
without any iterations.
Due to its noniterative nature,
our algorithm is hundreds of times faster than
$\ell_1$-norm minimization, and
methods based on expander graphs, both of which require multiple iterations.
Our algorithm can accommodate nearly sparse vectors, in which
case it recovers index set of the largest components, and can also
accommodate burst noise measurements.
Compared to compressive sensing methods that are guaranteed to achieve
exact recovery of \textit{all} sparse vectors,
our method requires fewer measurements.
However, methods that achieve \textit{statistical} recovery,
that is, recovery of \textit{almost} all but not all sparse vectors,
can require fewer measurements than our method.

% Numerical experiments with randomly generated sparse vectors indicate
% that the sufficient conditions for our algorithm to work are very close
% to being necessary.
% In contrast, the best known sufficient condition for $\ell_1$-norm minimization
% to recover a sparse vector, namely the Restricted Isometry Property 
% (RIP), is about thirty times away from being necessary. {\color{blue} This is consistent 
% with previous results which show that RIP-based measurement matrices even with
% the best known constants are not precise. For randomly constructed 
% matrices, the \textit{phase transition} of $\ell_1$-norm minimization determines the border between the failure and the success of $\ell_1$-norm minimization
% for various sparsities and under-sampling rates. It is previously shown that the phase transition
% lower bound provided by RIP are below the actual bound suggested by the
% polytope theory.
% Since we get the same conclusion by using deterministic measurement matrices
% instead of random matrices, we could think of this gap as relevant to
% the internal properties of the $\ell_1$-norm minimization algorithm.
% Other recovery algorithms like expander recovery and the non-iterative
% recovery proposed in this paper have by far smaller gaps between the
% sufficient and necessary conditions for the exact recovery.}
% {\color{red} Therefore it would be worthwhile to explore alternate and improved
% sufficient conditions for $\ell_1$-norm minimization to achieve
% the recovery of sparse vectors.}

\end{abstract}

%\begin{keywords}


%\end{keywords}

\section{Introduction}\label{sec:intro}

\subsection{Definition of Compressive Sensing}\label{ssec:11}

Compressive sensing refers to the recovery of 
high-dimensional but low-complexity objects 
using a very small number of measurements.
Some examples are the recovery of high-dimensional vectors that are
sparse (with very few nonzero components) or nearly sparse,
or high-dimensional matrices of low rank.
The focus in this paper is on vector recovery.
In order to make the discussion precise,
we give an exact formulation of what ``compressive
sensing'' means.
Our terminology more or less follows that in \cite{FR13,Cohen-Dahmen-Devore09}.
Throughout, $n$ denotes the dimension of the unknown vector, and $[n]$
denotes the set $\{ 1 , \ldots , n \}$.
The symbol $\supp(x) \seq [n]$ denotes the ``support'' of a vector $x \in \R^n$;
that is
\bd
\supp(x) := \{ i \in [n] : x_i \neq 0 \} .
\ed
If $k < n$ is a specified integer, then $\SI_k \seq \R^n$ denotes the set of
\textbf{$k$-sparse} vectors, that is
\bd
\SI_k := \{ x \in \R^n : | \supp(x) | \leq k \} .
\ed
% 
% \begin{definition}\label{def:sigma}
Suppose $\nmm{\cdot}$ is some specified norm on $\R^n$, and $k < n$
is a specified integer.
Then the \textbf{sparsity index} $\s_k(x,\nmm{\cdot})$ is defined by
\be\label{eq:11}
\s_k(x,\nmm{\cdot}) := \min_{z \in \SI_k} \nmm{x - z} .
\ee
% \end{definition}
For a given $x \in \R^n$ and an integer $k < n$, the symbols $x_d \in \R^n$
and $x_r \in \R^n$ denote respectively the \textbf{dominant} part 
and the \textbf{residual} part of $x$.
Thus $x_d$ is the vector consisting
of the $k$ largest components by magnitude of $x$ with the remaining components
set equal to zero, and $x_r = x - x_d$.
Note that, strictly speaking, we should write $x_{d,k}$ because the dominant
part depends on the specified integer $k$, but we do not do this in
the interests of less cluttered notation.
It is obvious that, for any $p \in [1,\infty]$, we have that
$\s_k(x,\nmp{\cdot}) = \nmp{x_r}$.
Also, in case of ``ties,'' when two or more components of $x$ have
the same magnitude, the symbol $x_d$ can be defined in any consistent fashion.

Following the notation in \cite{Cohen-Dahmen-Devore09}, we view compressive
sensing as consisting of two maps: A \textbf{measurement matrix}
$A \in \R^{mn}$ where $m < n$ is the number of measurements, and a
\textbf{decoding map} $\D: \R^m \ap \R^n$.

\begin{definition}\label{def:CS}
A pair $(A,\D)$ is said to achieve \textbf{exact sparse recovery} of order $k$
if
\be\label{eq:12}
\D(Ax) = x , \fa x \in \SI_k .
\ee
A pair $(A,\D)$ is said to achieve \textbf{stable sparse recovery} of order $k$
if there exists a constant $C$ such that
\be\label{eq:13}
\nmeu{ \D(Ax) - x } \leq C \s_k(x,\nmm{\cdot}_1) , \fa x \in \R^n .
\ee
A pair $(A,\D)$ is said to achieve \textbf{robust sparse recovery} of order $k$
if there exist constants $C, D$ such that, whenever $\eta \in \R^m$
satisfies $\nmeu{\eta} \leq \e$, we have that
\be\label{eq:14}
\nmeu{ \D(Ax + \eta) - x} \leq C \s_k(x,\nmm{\cdot}_1) + D \e ,
\fa x \in \R^n .
\ee
\end{definition}

It is easy to verify that robust recovery implies stable recovery, which
in turn implies exact recovery.

Note that in Definition \ref{def:CS}, the various inequalities are
required to hold for \textit{all} vectors $x$.
For want of a better phrase, this could be thought of as ``guaranteed''
sparse recovery.
An alternate and weaker requirement would be ``statistical'' sparse
recovery, in which the pair $(A,\D)$ recovers almost all,
but not necessarily all, 
sparse vectors, with respect to a predefined probability measure on the set of
vectors.
In such a case the number of measurements $m$ can be significantly reduced.
We will describe this alternative in greater detail in Section 
\ref{ssec:23}.

\subsection{Overview and Our Contributions}\label{ssec:12}

By now there are several approaches to the recovery of sparse vectors.
If $x \in \R^n$ is an unknown sparse vector and $y = Ax \in \R^n$ is
the measured vector, then the most ``logical'' approach to finding $x$
would be to solve
\bd
\xh = \argmin_{z} \nmm{z}_0 \st Az = y ,
\ed
where $\nmm{z}_0 = | \supp (z) |$ denotes the number of nonzero
components of $z$.
Unfortunately this problem is NP-hard, as shown in \cite{Natarajan95}.
Therefore one can think of replacing the function $\nmm{\cdot}_0$ by
its ``convex envelope,'' which is the largest convex function that is
dominated by $\nmm{\cdot}_0$.
Using methods similar to those in \cite{Fazel-Hindi-Boyd01},
it can be shown that the convex envelope of $\nmm{\cdot}_0$ 
over the unit ball in $\nmm{\cdot}_\infty$ is $\nmm{\cdot}_1$.
% In \cite{Needell-Tropp09} they are grouped into three main categories,
% namely \textit{convex relaxation}, \textit{greedy pursuit},
% and \textit{combinatorial algorithms}.
Consequently the most popular algorithm is $\ell_1$-norm minimization
with a measurement matrix $A$ chosen to satisfy the so-called
restricted isometry property (RIP), which is defined precisely in
Section \ref{ssec:21}.
The methods used for choosing the matrix $A$ so as to satisfy the RIP
can either be deterministic or probabilistic.
Probabilistic methods lead to $A$ matrices that do not have any structure;
moreover, verifying whether such a probabilistically generated
matrix satisfies the RIP condition
% (a widely-used sufficient condition)
is NP-hard \cite{BDMS13}.
In contrast, deterministically constructed matrices
are often binary and thus readily
verified to have the required properties, easy to implement, and
faster than using random matrices.
However, even with binary measurement matrices,
convex optimization is far slower than greedy methods
such as matching pursuit \cite{Mallat-Zhang-MP}, orthogonal matching
pursuit \cite{PRK-OMP,Tropp-TIT04}, and CoSaMP
\cite{Needell-Tropp09}.
The main disadvantage of greedy methods is that the known sufficient
conditions for them to work are more stringent than those for convex
relaxation methods to work.
Another recent innovation, that combines the advantages of greedy
algorithms with weak sufficient conditions of convex optimization,
is based on the use of expander graphs \cite{Xu-Hassibi07,JXHC09}.
The measurement matrices in this approach are always binary and thus
easy to implement.
Further, the measurement matrices in this approach are expected to
satisfy an analog of the RIP condition, known as $\ell_1$-RIP.
There is reason to believe that, at least in principle, the number
of measurements in this approach can be order-optimal
\cite{Indyk-Ruzic08,Indyk-Allerton08b}.

The above discussion pertains to methods that lead to the recovery of
\textit{all} sufficiently sparse vectors.
By weakening the requirement, and asking only that the algorithm
be able to recover \textit{almost} all sufficiently sparse vectors,
the number of measurements drops drastically.
Among such methods, approximate message passing (AMP) introduced in
\cite{Donoho-et-al-PNAS09} is among the most thoroughly studied.
The analysis in \cite{Donoho-et-al-PNAS09} shows that the recovery
algorithm displays a \textit{phase transition} whereby the performance
of the algorithm undergoes an abrupt change.
A readable survey of these results is given in \cite{Donoho-Tanner10}.
In \cite{Wu-Verdu12}, a very general model is studied, wherein the
unknown vector obeys a known probability distribution, the encoder
is allowed to be nonlinear (as opposed to a linear map $A$), and the decoder
is Lipschitz-continuous and may make use of the known probability 
distribution of the unknown vector.
It is shown that, in the case of i.i.d. input processes, 
the phase-transition threshold for optimal encoding
is given in terms of the R\'{e}nyi 
information dimension of the input distribution.
More details are given in Section 
\ref{ssec:23}.
In \cite{Donoho-et-al-TIT13b}, the procedure in
\cite{Wu-Verdu12} is incorporated into the approximate message passing
(AMP) framework, along with the idea of spatial coupling taken from
\cite{Krzakala-et-al12},
and the phase transition properties of this algorithm
are studied.
In \cite{Amelunxen-et-al14}, a general theory of phase transitions
is developed for convex regularizers, which includes the widely used
technique of $\ell_1$-norm minimization, and a very sharp 
characterization of the phase transition boundary is given.
Thus there is a fairly complete theory for the case of ``statistical'' signal
recovery, which is only briefly touched upon here.

Against this backdrop, we now describe the contributions of the present
paper and place them in perspective.
In this paper we present a new \textit{noniterative} algorithm for
the recovery of vectors that are extremely sparse.
Due to its noniterative nature, it is hundreds of times faster than 
both $\ell_1$-norm minimization using binary measurement matrices,
and methods based on expander graphs.
Our algorithm works also for vectors that are nearly sparse but not 
exactly sparse, and/or measurements that are corrupted by noise.
% In terms of \textit{known sufficient conditions} to work, our new
Our method requires slightly fewer measurements than either the RIP condition or
the expander graph construction.
However, this claim should be tempered by the results based on 
approximate message passing and the R\'{e}nyi information dimension, which
show that if one is ready to settle for the recovery of \textit{almost all}
sparse vectors, the required number of measurements is far smaller than
those required by the RIP condition, and by inference, far smaller than
those required by our method.

\section{Literature Review}\label{sec:review}

\subsection{Compressive Sensing via $\ell_1$-Norm Minimization}\label{ssec:21}

One of the most popular approaches to compressive sensing is $\ell_1$-norm
minimization.
This approach is originally introduced in
\cite{Chen-Donoho-Saunders99,Chen-Donoho-Saunders01} as a heuristic
and is called ``basis pursuit.''
The method consists of defining the decoding map $\D$ as
\be\label{eq:15}
\D(y) = \xh := \argmin_z \nmm{z}_1 \st y = Az
\ee
in the case where $y = Ax$, and
% of noiseless measurements, and
\be\label{eq:15a}
\D(y) = \xh := \argmin_z \nmm{z}_1 \st \nmeu{ y - Az } \leq \e 
\ee
in the case where $y = Ax + \eta$ with $\nmeu{\eta} \leq \e$.
% of noisy measurements.
In several papers over the years, beginning with
\cite{Candes-Tao05,Donoho06b},
it is shown that $\ell_1$-norm minimization achieves robust sparse
recovery in the sense of Definition \ref{def:CS}
if the matrix $A$ satisfies suitable conditions.
At present, the two most popular
sufficient conditions for $\ell_1$-norm minimization to achieve robust
sparse recovery are the restricted isometry property (RIP), and
the robust null space property (RNSP).
It is shown in \cite{MV-Ranjan17} that RIP implies the RNSP, so
we restrict our attention to the RIP.

% The restricted isometry property, introduced in \cite{Candes-Tao05}, is
% defined next.

\begin{definition}\label{def:RIP}
A matrix $A \in \R^{mn}$ is said to satisfy the \textbf{restricted isometry
property (RIP)} of order $k$ with constant $\d$ if
\be\label{eq:16}
(1 - \d) \nmeusq{u} \leq \nmeusq{Au} \leq (1 + \d) \nmeusq{u} ,
\fa u \in \SI_k .
\ee
\end{definition}

Available results show that $\ell_1$-norm minimization achieves
robust sparse 
recovery provided the measurement matrix $A$ satisfies the RIP with
a sufficiently small constant.
The definitive results in this direction are derived in \cite{CZ14}.

\begin{theorem}\label{thm:Cai-Zhang-1}
(See Theorems 1.1 and 2.1 of \cite{CZ14}.)
Suppose that, for some number $t > 1$, the matrix $A$ satisfies the RIP
of order $\lceil tk \rceil$ with constant $\d < \sqrt{(t-1)/t}$.
Then $\ell_1$-norm minimization as in \eqref{eq:15a} achieves robust
sparse recovery.
% For every $t \geq 4/3$, the above bound $\d < \sqrt{(t-1)/t}$ is tight.
\end{theorem}

\begin{theorem}\label{Cai-Zhang-2}
(See \cite[Theorem 2.2]{CZ14}.)
Suppose $t \geq 4/3$.
Then for all $\xi > 0$ and all $k \geq 5/\xi$, there exists a matrix $A$
that satisfies the RIP of order $tk$ with constant
$\d_{tk} < \sqrt{(t-1)/t} + \xi$, and a vector $x \in \SI_k$ such that
\ben
\item With the noise-free measurement $y = Ax$, the decoder map
$\D$ defined in \eqref{eq:15} fails to recover $x$.
\item With a noisy measurement $y = Ax + \eta$ where $\nmeu{\eta} \leq \e$,
the decoder map $\D$ defined in \eqref{eq:15a} fails to recover $x$.
\een
\end{theorem}

This raises the question as to how one may construct measurement matrices 
that satisfy the RIP.
There are two distinct approaches to this, namely probabilistic,
and deterministic.
In the probabilistic approach, $A$ is chosen to equal $(1/\sqrt{m}) \Phi$,
where $\Phi$ is an $m \times n$ matrix consisting of independent samples
of a sub-Gaussian random variable, such as Bernoulli or any finite-valued
random variable, or a normal random variable.
Such a construction leads to a matrix $A$ that satisfies the RIP
\textit{with high probability} which can be made close to, but not equal to,
one.
Moreover, as shown in \cite{BDMS13}, verifying whether \textit{a particular}
randomly generated matrix $A$ satisfies the RIP is NP-hard.
An alternative that is gathering interest in recent times is the use
of deterministic methods.
The paper \cite{DeVore07} is apparently the first to provide a
deterministic method for constructing matrices that satisfy the RIP.
This construction results in a binary matrix that is well-suited
for implementation.

\subsection{Compressive Sensing Using Expander Graphs}\label{ssec:22}

A recent development is the application of ideas from algebraic coding
theory to compressive sensing.
In \cite{SBB06}, a method called ``sudo-codes'' is proposed, which is
based on low density parity check (LDPC) codes, which are well-established
in coding theory.
The sudo-codes method can recover sparse signals with high probability.
Motivated by this method, Xu and Hassibi in \cite{Xu-Hassibi07} proposed
a method based on expander graphs, 
which are a special type of bipartite graph.
% The measurement matrix used in \cite{Xu-Hassibi07} is the biadjacency
% matrix of an expander graph with an expansion factor of $3/4$ or more,
% and is therefore binary.
For the convenience of the reader, the definition of an
expander graph is recalled next.

The object under study is an
\textit{undirected bipartite graph}, consisting of a set
$\V_I$ of input vertices, a set $\V_O$ of output vertices, and an edge set
$\E \seq \V_O \times \V_I$, where $(i,j) \in \E$ if and only if there is
an edge between node $i \in \V_O$ and node $j \in \V_I$.
The corresponding matrix $A \in \bi^{|\V_O| \times |\V_I|}$ is called
the \textbf{bi-adjacency matrix} of the bipartite graph.
The graph is said to be \textbf{left-regular} of degree $D$,
or \textbf{$D$-left regular}, if every input node has degree $D$.
This is equivalent to requiring that every column of the bi-adjacency matrix
$A$ has exactly $D$ elements equal to $1$.
Given an input vertex $j \in \V_I$, let $\N(i) \seq \V_O$ denote the
set of its neighbors, defined as
\bd
\N(j) := \{ i \in \V_O : (i,j) \in \E \} .
\ed
Given set of input vertices $S \seq \V_I$, the set of its neighbors
$\N(S) \seq \V_O$ is defined as
\bd
\N(S) := \bigcup_{j \in S} \N(j) 
= \{ i \in \V_O : \exists j \in S \st (i,j) \in \E \} .
\ed

\begin{definition}\label{def:exp}
A $D$-left regular bipartite graph $(\V_I,\V_O,\E)$ is said to be a
\textbf{$(K,1 - \beta)$-expander} for some integer $K$ and some
number $\beta \in (0,1)$ if, for every $S \seq V_I$ with $|S| \leq K$,
we have that $| N(S) | \geq (1 - \beta) D |S|$.
\end{definition}

It can be shown \cite{Capalbo-et-al02} that randomly generated left-regular
graphs are expanders.
The next theorem is a paraphrase of \cite[Theorem 13.6]{FR13} in the
current notation; note that this theorem is based on
\cite{Capalbo-et-al02}.

\begin{theorem}\label{thm:rand-exp}
Given integers $d, m, n$ with $d < m < n$, let $\B(m,n,d)$ denote the
set of $d$-left-regular bipartite graphs.
Suppose an integer $K < m/d$ and real numbers $\beta \in (0,1)$,
$\e \in (0,0.5)$ are specified.
Define
\be\label{eq:6201}
d = \left\lceil \frac{1}{\beta} \ln \left( \frac{en}{2 \e} \right) \right\rceil ,
m = \lceil \exp(2/\beta) dK \rceil ,
\ee
where $e$ denotes the base of the natural logarithm.
Then the fraction of graphs in $\B(m,n,d)$ that are $(K,1-\beta)$ expanders
exceeds $1 - \e$.
\end{theorem}

Note that the above theorem is not very useful, because testing whether
\textit{a given} randomly generated left-regular graph is a
$(K,1-\beta)$-expander or not would require us to compute the neighbors of
$\left( \ba{c} n \\ K \ea \right)$ sets of vertices.
While this number is ``polynomial'' in $n$, it would be impractically large.

In \cite{Xu-Hassibi07}, Xu and Hassibi introduce a new signal recovery algorithm
in which the bi-adjacency matrix of an expander graph with $\beta \leq 1/4$
is used as the measurement matrix.
It is referred to here as the ``Expander Recovery Algorithm.''
Xu and Hassibi show that their algorithm recovers an unknown $k$-sparse
vector $x$ \textit{exactly} in $O(k \log n)$ iterations.
Subsequently, their method was updated in \cite{JXHC09}
by increasing the expansion factor from
$1 - 1/4 = 3/4$ to $1-\e$ in which $\e < 1/4$.
With this change, it is shown that the number of recovery iterations
required is $O(k)$.
However, the number of measurements is more than in the Xu-Hassibi algorithm.

\begin{algorithm}{\bfseries{Expander Recovery Algorithm}}
\begin{algorithmic}[1]
\State Initialize $x-0_{n\times 1}$
\If{ $Y=Ax$} \Return {output $x$ and exit}
\Else \\
find a variable $x_j$ such that at least $(1-2\epsilon)$D
of the measurements it participates in have identical gap $g$
\EndIf
\State{$x_j \gets x_j+g$ and go to step 2}
\State{\bfseries end if}
\end{algorithmic}
\label{alg:Xu-Has}
\end{algorithm}
In the algorithm above, the term $g$ is called the \textit{gap} and it determines the amount of information of the unknown signal that is missing in the estimate. The gap is defined as following:
$$g_i=y_i-\sum_{j=1}^{n} {A_{ij}x_j}$$
in which $x,y$ and $A$ are the unknown signal, the measurement vector and the 
measurement matrix, respectively.

\subsection{Statistical Recovery and Phase Transitions}\label{ssec:23}

The preceding two subsections were devoted to two methods for 
\textit{guaranteed} recovery of all sparse vectors.
For such methods, the restricted isometry property (RIP) is the most
popular sufficient condition.
It is known that, if the measurement matrix $A$ is generated in a
probabilistic fashion, the RIP holds with high probability with
$m = O(k \ln (n/k))$.
In contrast, with deterministic methods, the number of measurements $m$
is typically $O(\max\{k^2,n^{1/2}\})$.
However, in practice deterministic methods often require fewer measurements.
Further details can be found in the Appendix.

If the requirement is relaxed from \textit{guaranteed} recovery
to recovery with high probability, or \textit{statistical} recovery,
then there is a parallel body of research showing that the number of
measurements $m$ can be reduced quite substantially.
In fact $m = O(k)$ measurements suffice.
In this subsection, we highlight just a few of the many papers in this
area of research.
To streamline the presentation, the papers are not always cited in
chronological order.

In \cite{Wu-Verdu12}, the underlying assumption is that the unknown
vector $x$ is generated according to a known probability distribution
$p_X$, which can in fact be used by the decoder.
Three different dimensions of the probability distribution $p_X$ are
introduced, namely the R\'{e}nyi information dimension, the MMSE
dimension, and the Minkowski dimension.
The encoder is permitted to be nonlinear, in contrast to earlier cases
where the encoding consisted of multiplication by a measurement matrix.
The decoder is also permitted to be nonlinear but is assumed to be Lipschitz
continuous.
The optimal performance in this setting is analyzed.
A central result in this paper states that, asymptotically as the vector
dimension $n$ and the number of measurements $m$ both approach infinity,
statistical recovery is possible \textit{if and only if}
\bd
m \geq n \bar{d}(p_X) + o(n) ,
\ed
where $\bar{d}(p_X)$ denotes the R\'{e}nyi information dimension of $p_X$.
Since the R\'{e}nyi information dimension is comparable to the ratio $k/n$,
the above result states that $O(k)$ measurements are sufficient.
However, no procedure is given to construct an encoder-decoder pair.

In a series of papers \cite{Donoho06c,Donoho-et-al-PNAS09,Donoho-Tanner-JAMS09},
Donoho and various co-workers studied ``phase transitions'' in the performance
of various recovery algorithms.
A readable survey of these results is given in \cite{Donoho-Tanner10}.
The unknown $n$-vector is assumed to be $k$-sparse, and the measurement
vector $y \in \R^m$ equals $Ax$, where $A$ consists of samples of normal
random variables, scaled by the normalization factor $1/\sqrt{m}$.
Two quantities are relevant here, namely the ``under-sampling rate''
$\d = m/n$, and the sparsity $\r = k/m$.
In all of these papers, the aim is to show that for each algorithm
there exists a sharp
threshold $\r_\th(\d)$ such that, if $\r > \r_\th(\d)$, then the unknown
vector is recovered with probability approaching one, whereas if 
$\r < \r_\th(\d)$, then the algorithm \textit{fails} with probability
approaching one.

Specifically in \cite{Donoho-et-al-PNAS09} an algorithm known as
``approximate message passing'' (AMP) is analyzed.
AMP is a simple thresholding type of algorithm that is much faster than
minimizing the $\ell_1$-norm.
Specifically, suppose $\phi : \R \ap \R$ is a smooth ``threshold''
function, and extend it to a map from $\R^n$ to $\R^n$ by applying it
component-wise.
The AMP algorithm begins with an initial guess $x^0 = 0$, and then one sets
\bd
x^{t+1} = \phi(A^\top w^t + x^t ) ,
\ed
\bd
w^t = y - Ax^t + \frac{1}{\d} w^{t-1} ( \phi'( A^\top w^{t-1} + x^{t-1} ) ) ,
\ed
where $\phi'$ denotes the derivative of $\phi$.
It is clear that AMP is much faster than $\ell_1$-norm minimization.
Despite this, it is shown in \cite{Donoho-et-al-PNAS09} that
the phase transition behavior of AMP is comparable to that of $\ell_1$-norm
minimization.
In \cite{Donoho-et-al-TIT13b}, the AMP algorithm is modified to incorporate
the results of \cite{Wu-Verdu12}, and phase transition results are derived.
In this paper, the authors also introduce the idea of ``spatial coupling''
introduced in \cite{Krzakala-et-al12}.

Finally, in \cite{Amelunxen-et-al14}, the authors study a very general
class of algorithms.
Suppose as before that $y \in \R^m$ equals $Ax$, where $A$ consists of
samples of normal random variables, scaled by the normalization factor
$a/\sqrt{m}$.
The decoding algorithm is
\bd
\xh = \argmin_z f(z) \st y = Az ,
\ed
where the ``regularizer'' $f(\cdot)$ is a convex function satisfying
some technical conditions.
So this theory applies to $\ell_1$-norm minimization.
In this paper, a central role is played by the ``descent cone'' of $f$
at a point $x$, which is defined as
\bd
\Dcal(f,x) := \bigcup_{\t > 0} \{ h \in \R^n : f(x+\t h) \leq f(x) \} .
\ed
It is clear that $\Dcal(f,x)$ is indeed a cone.
Next, for each cone, a quantity called the ``statistical dimension,''
denoted by $\d$, is defined; see \cite[Section 2.2]{Amelunxen-et-al14}
for a precise definition.
With all these items in place, a central result is established;
see \cite[Theorem II]{Amelunxen-et-al14}.

\begin{theorem}
Define $a(\e) := \sqrt{8 \log(4/\e)}$.
With all other symbols as above, if
\bd
m \leq \d(\Dcal(f,x)) - a(\e) \sqrt{n} ,
\ed
then the decoding algorithm fails with probability $\geq 1 - \e$.
If
\bd
m \geq \d(\Dcal(f,x)) + a(\e) \sqrt{n} ,
\ed
then the decoding algorithm succeeds with probability $\geq 1 - \e$.
\end{theorem}

\section{The New Algorithm}\label{sec:new}

Now we present our new algorithm, and show that it can exactly
recover sparse signals in a single pass, without any iterations.
Then we analyze the performance of the algorithm when the true but unknown
vector is not exactly sparse, and/or the measurement is corrupted by noise.
The performance of our algorithm is compared with those of $\ell_1$-norm
minimization and expander graph algorithms in the next section.

\subsection{The New Algorithm}

Suppose a matrix $A \in \bi^{m \times n}$ has the following properties,
referred to as the \textbf{main assumption}:
\ben
\item Every column $a_j$ of $A$ has precisely $q$ entries of $1$
and $m-q$ entries of $0$.
\item If $a_j,a_t$ are distinct columns of $A$, then
$\IP{a_j}{a_t} \leq r-1$.
\een
Suppose $x \in \SI_k$ is a $k$-sparse $n$-dimensional vector,
and define $y = Ax$ to be the measurement vector.
For a given index $j \in [n]$, let $\{ v_1(j) , \ldots , v_q(j) \}
\seq [m]$ denote the $q$ rows such that $a_{ij} = 1$.
For an index $j \in [n]$,
the \textbf{reduced measurement vector} $\bar{y}_j \in \mathbb R^{q}$ is defined as
\bd
\bar{y}_j := [ y_{v_1(j)} \ldots y_{v_q(j)} ]^\top .
\ed
Note that $\yb_j$ is the vector consisting of the $q$ measurements
in which the component $x_j$ participates.

The main result is given next.
Recall that $\nmm{v}_0$ denotes the number of
nonzero components of a vector $v$.

\begin{theorem}\label{thm:main-1}
Suppose $x \in \SI_k$, $y = Ax$.
Then:
\ben
\item If $j \not \in \supp(x)$, then
$\nmm{ \bar{y}_j }_0 \leq k(r-1)$.
\item If $j \in \supp(x)$, then $\bar{y}_j$ contains at least
$q - (k-1)(r-1)$ components that are all equal to $x_j$.
\een
\end{theorem}

\begin{proof}
For $t \in [n]$, let $\eb_t \in \R^n$ denote the $t$-th canonical basis
vector, which has a $1$ as its $t$-th element, and zeros elsewhere,
and let $\oneb_q \in \R^q$ denote the column vector consisting of
all ones.
Then we can write:
\bd
x = \sum_{t \in \supp(x)} x_t \eb_t ,
\ed
\bd
y = Ax = \sum_{t \in \supp(x)} x_t A \eb_t
= \sum_{t \in \supp(x)} x_t a_t ,
\ed
where $a_t$ denotes the $t$-th column of $A$.
Therefore, for a fixed $j \in [n]$ and $l \in [q]$, we have that
\bd
y_{v_l(j)} = \sum_{t \in \supp(x)} x_t (a_t)_{v_l(j)} .
\ed
Letting $l$ range over $[q]$ shows that
\be\label{eq:21}
\bar{y}_j = \sum_{t \in \supp(x)} x_t (\overline{a_t})_j ,
\ee
where $(\overline{a_t})_j$ is the reduced vector of $a_t$ consisting of
$(a_t)_{v_1(j)} , \ldots , (a_t)_{v_q(j)}$.

\textbf{Proof of (1):}
Suppose $j \not \in \supp(x)$.
Then $j \neq t$ for all $t \in \supp(x)$.
Therefore, according to item (ii) of the main assumption, we have that
$\IP{a_j}{a_t} \leq r-1$.
Recall that $v_1(j) , \ldots , v_q(j)$ are the row indices of column $j$
that contain a $1$.
Therefore, for a fixed index $t \neq j$, the number of $1$'s in the
set $\{ (a_t)_{v_1(j)} , \ldots , (a_t)_{v_q(j)} \}$ equals the inner product
$\IP{a_j}{a_t}$ and thus cannot exceed $r-1$.
Therefore, for a fixed index $t \in \supp(x)$, the vector
$x_t (\overline{a_t})_j$ contains no more than $r-1$ nonzero entries.
Substituting this fact into \eqref{eq:21} shows that $\bar{y}_j$ is the sum
of at most $k$ vectors (because $x$ is $k$-sparse), each of which
has no more than $r-1$ nonzero entries.
Therefore $\nmm{ \bar{y}_j }_0 \leq k(r-1)$.

\textbf{Proof of (2):}
Suppose $j \in \supp(x)$.
Then we can write
\beq
\bar{y}_j & = & \sum_{t \in \supp(x)} x_t (\overline{a_t})_j \\
& = & x_j \oneb_q + \sum_{t \in \supp(x) \setminus \{ j \} }
x_t (\overline{a_t})_j ,
\label{eq:22}
\eeq
because the ``reduced vector'' $(\overline{a_j})_j$ consists of $q$ $1$'s,
as denoted by $\oneb_q$.
By the same reasoning as in the proof of (1), it follows that
\bd
\left\nm \sum_{t \in \supp(x) \setminus \{ j \} }
x_t (\overline{a_t})_j \right\nm_0 \leq (k-1)(r-1) .
\ed
Therefore at least $q - (k-1)(r-1)$ terms in $\bar{y}_j$ equal $x_j$.
\end{proof}

In view of Theorem \ref{thm:main-1}, we can formulate an algorithm for
the recovery of $k$-sparse vectors, as follows:

\begin{algorithm}{\bfseries{New Recovery Algorithm}}
\begin{algorithmic}[1]
% \State Initialize $x=0_{n\times 1}, \delta=10^{(-5)} (Threshold), Y\in  \mathbb R^m (
% measurement\ vector\ given)$
\For{$j \in [n]$}
% \State Construct $V_j=v_{l}(j), l\in [1,...,q]$, which includes the row indices of $A$ 
% which is equal to one
\State Construct the \textit{reduced measurement vector} $\bar{y}_{j}$.
\State Find the number of the elements of $\bar{y}_{j}$ that are
nonzero; call it $\nu$.
\Comment{(In implementation, we find the number of elements that are greater
than some tolerance $\delta$.)}
\If{$\nu > q/2$}
\State Find a group of $q/2$ elements in $\bar{y}_j$ that are equal;
call this value $\th_j$.
\Comment{(In implementation, we allow some tolerance here.)}
\State $\xh_j = \th_j$.
\Else
\State $\xh_j=0$
\EndIf 
\State {\bfseries end}  
\EndFor
\State {\bfseries end}
\end{algorithmic}
\end{algorithm}

% \textbf{Algorithm:} Given $y = Ax \in \R^m$, where the matrix $A$
% satisfies the main assumption, $q > 2k(r-1)$.
% For each $j \in [n]$, compute the reduced vector $\yb_j$.
% If $\nmm{\yb_j}_0 < q/2$, deduce that $j \not \in \supp(x)$,
% and set $\xh_j = 0$.
% If $\nmm{\yb_j}_0 > q/2$, identify a number $v_j$ such that
% at least $q/2$ elements of $\yb_j$ equal $v_j$, and set $\xh_j = v_j$.

Note that there is no iterative process involved in the recovery --
% one just ``reads off'' the various entries of $x$.
the estimate $\xh$ is generated after \textit{a single pass}
through all $n$ indices.

\begin{theorem}\label{thm:main-2}
If $x$ is $k$-sparse, and $A$ satisfies the main assumption
with $q > 2k(r-1)$, then $\xh = x$.
% the above algorithm recovers $x$ in that $\xh = x$.
\end{theorem}

\begin{proof}
Note $q > 2k(r-1)$ implies that
\bd
k(r-1) < q/2, q - (k-1)(r-1) > q - k(r-1) > q/2 .
\ed
Therefore, by Statement 1 of Theorem \ref{thm:main-1}, it follows that if
$j \not \in \supp(x)$, then $\nmm{\bar{y}_j}_0 \leq k(r-1) < q/2$.
Taking the contrapositive shows that if $\nmm{\bar{y}_j}_0 \geq q/2$,
then $j \in \supp(x)$.
Therefore, by Statement 2 of Theorem \ref{thm:main-1}, it follows that
at least $q - (k-1)(r-1) > q - k(r-1) > q/2$ elements of $\bar{y}_j$ equal $x_j$.
\end{proof}

% \section{Extension to Noisy Measurements And/Or Non-Sparse Vectors}

Next we present the extension of our basic algorithm to
the cases of a sparse signal with measurement noise, and a nearly sparse
signal.
% Due to space limitations, the theorems are stated without proof.

\subsection{Recovery of Sparse Signals with Measurement Noise}

In previous work, the model for noisy measurements
is that $y = Ax + \eta$ where there
is a prior bound of the form $\nmeu{\eta} \leq \e$.
If $x \in \SI_k$, then $\s_k(x,\nmm{\cdot}_1) = 0$.
Therefore, if robust sparse recovery is achieved, then
the bound in \eqref{eq:14} becomes
% The recovered vector $\xh$ satisfies a bound of the form
$\nmeu{\xh - x} \leq D \e$.
However, our approach draws its inspiration from coding theory, wherein it
is possible to recover a transmitted signal correctly provided
the transmission is not corrupted in too many places.
Therefore our noise model is that $\nmm{\eta}_0 \leq M$.
In other words, it is assumed that a maximum of $M$ components of the
``true'' measurement $Ax$ are corrupted by additive noise, but there
are no assumptions regarding the \textit{magnitude} of the error signal $\eta$.
In this case it is shown that, by increasing the number of measurements,
it is possible to recover the true sparse vector $x$ \textit{perfectly}.

\begin{theorem}\label{thm:noisy}
Suppose $x \in \SI_k$, and that $y = Ax + \eta$ where $\nmm{\eta}_0 \leq M$.
Suppose further that the matrix $A$ satisfies the main assumption.
Then
\ben
\item If $j \not \in \supp(x)$, then $\yb_j$ contains no more than
$k(r-1) + M$ nonzero components.
\item If $j \in \supp(x)$, then $\yb_j$ contains at least
$q - [(k-1)(r-1) + M]$ components that are all equal to $x_j$.
\item Suppose the new recovery algorithm is applied with a measurement
matrix $A$ that satisfies the main assumption with $q > 2[k(r-1)+M]$.
Then $\xh = x$.
\een
\end{theorem}

\begin{proof}
Suppose $x \in \SI_k$ and let $y = Ax + \eta$ where $A$ satisfies the main
assumption and $\nmm{\eta}_0 \leq M$.
Let $u = Ax$ denote the uncorrupted measurement.
For a fixed index $j \in [n]$, let $\yb_j \in \R^q$ denote the reduced
measurement vector, consisting of the components $y_{v_1(j)}$ through
$y_{v_q(j)}$, and define $\bar{u}_j \in \R^q$ and $\bar{\eta}_j \in \R^q$
analogously.

First suppose $j \not \in \supp(x)$.
Then it follows from Item (1) of Theorem \ref{thm:main-1} that
$\nmm{\bar{u}_j}_0 \leq k(r-1)$.
Moreover, because $\eta$ has no more than $M$ nonzero components and
$\bar{\eta}_j$ is a sub-vector of $\eta$, it follows that
$\nmm{\bar{\eta}_j}_0 \leq M$.
Therefore
\bd
\nmm{\yb_j}_0 = \nmm{ \bar{u}_j + \bar{\eta}_j }_0
\leq \nmm{ \bar{u}_j }_0 + \nmm{ \bar{\eta}_j }_0 \leq k(r-1) + M .
\ed
This is Item (1) above.
Next, suppose that $j \in \supp(x)$.
Then it follows from Item (1) of Theorem \ref{thm:main-1} that
at least $q - (k-1)(r-1)$ elements of $\bar{u}_j$ equal $x_j$.
Because $\nmm{\bar{\eta}_j}_0 \leq M$, it follows that at least
$q - (k-1)(r-1) - M$ components of $\yb_j$ equal $x_j$.
This is Item (2) above.
Finally, if $q > 2k(r-1) + 2M$, it follows as in the proof of Theorem
\ref{thm:main-2} that $\xh = x$.
\end{proof}

Note that the assumption on the noise signal $\eta$ can be modified
to $\nmm{ \bar{\eta}_j }_0 \leq M$ for each $j \in [n]$.
In other words, instead of assuming that $\eta$ has no more than $M$
nonzero components, one can assume that every reduced vector
$\bar{\eta}_j$ has no more than $M$ nonzero components.

\subsection{Recovery of Nearly Sparse Signals}

% Suppose $x \in \R^n$, and let $k$ be a specified integer.
As before, if $x \not \in \SI_k$, then let $x_d \in \R^n$ denote the
projection of $x$ onto its $k$ largest components, and let $x_r = x - x_d$.
We refer to $x_d, x_r$ as the dominant part and the residual respectively.
Note that, for any $p \in [1,\infty]$, we have that the sparsity index
$\s_k(x,\nmp{\cdot})$ equals $\nmp{x_r}$.
To (nearly) recover such a vector, we modify the New Recovery Algorithm
slightly.
Let $\d$ be a specified threshold.

\begin{algorithm}{\bfseries{Modified Recovery Algorithm}}
\begin{algorithmic}[1]
\For{$j \in [n]$}
\State Construct the \textit{reduced measurement vector} $\bar{y}_{j}$.
\State Find the number of the elements of $\bar{y}_{j}$ that are
greater than $\d$ in magnitude; call it $\nu$.
% \Comment{(In implementation, we find the number of elements that are greater
% than some tolerance $\delta$.)}
\If{$\nu > q/2$}
\State Find a group of $q/2$ elements in $\bar{y}_j$ such that the
difference between the largest and smallest elements is no larger than $2\d$;
Let $\th_j$ denote the average of these numbers.
% \Comment{(In implementation, we allow some tolerance here.)}
\State $\xh_j = \th_j$.
\Else
\State $\xh_j=0$
\EndIf
\State {\bfseries end}
\EndFor
\State {\bfseries end}
\end{algorithmic}
\end{algorithm}

\begin{theorem}\label{thm:non-sparse}
Suppose $x \in \R^n$ and that $\s_k(x,\nmm{\cdot}_1) \leq \d$.
Write $x = x_d + x_r$ where $x_d$ is the dominant part of $x$ consisting
of its $k$ largest components, and $x_r = x - x_d$ is the residual.
Let $y = Ax$ where $A$ satisfies the main assumption with $q > 2k(r-1)$, and
apply the modified recovery algorithm.
Then (i) $\supp(\xh) = \supp(x_d)$, and (ii) $\nmm{\xh - x_d}_\infty \leq \d$.
\end{theorem}

\textbf{Remark:}
If $\ell_1$-norm minimization is used to recover a \textit{nearly sparse}
vector using \eqref{eq:15},
then the resulting estimate $\xh$
need not be sparse, and second, the support set of the dominant
part of $\xh$ need not equal the support
set of the dominant part of $x$.

\begin{proof}
Write $x = x_d + x_r$ where $x_d$ consists of the dominant part of $x$ and 
$x_r$ consists of the residual part.
By assumption, $\nmm{x_r}_1 \leq \d$.
Note that the measurement $y$ equals $Ax = A x_d + A x_r$.
Let $u = A x_d$ and observe that $x_d \in \SI_k$.
Further, observe that, because the matrix $A$ is binary, we have that
the induced matrix norm
\bd
\nmm{A}_{1 \ap \infty} := \sup_{v \neq 0} \frac{ \nmm{Av}_\infty }
{ \nmm{v}_1 } = \max_{i,j} | a_{ij} | = 1 .
\ed
Therefore $\nmm{A x_r}_\infty \leq \nmm{x_r}_1 \leq \d$.
Now, by Item (1) of Theorem \ref{thm:main-1}, we know that if
$j \not \in \supp(x_d)$, then no more than $k(r-1)$ components of 
the reduced vector $\bar{u}_j$ are nonzero.
Therefore then no more than $k(r-1)$ components of
the reduced vector $\yb_j$ have magnitude more than $\d$.
By Item (2) of Theorem \ref{thm:main-1}, we know that if
$j \in \supp(x_d)$, then at least $q - (k-1)(r-1)$ components of $\bar{u}_j$
equal $x_j$.
Therefore at least $q - (k-1)(r-1)$ components of $\yb_j$ lie in the
interval $[x_j - \d , x_j + \d ]$.
Finally, if $q > 2k(r-1)$, then there is only one collection of
$q - (k-1)(r-1) > q/2$ components of the reduced vector $\yb_j$ that
lie in an interval of width $2\d$.
The true $x_j$ lies somewhere within this interval, and we can set $\xh_j$
equal to the midpoint of the interval containing all of these components.
In this case $| \xh_j - x_j | \leq \d$.
Because this is true for all $j \in \supp(x_d)$, it follows that
(i) $\supp(\xh) = \supp(x_d)$, and (ii) $\nmm{\xh - x_d}_\infty \leq \d$.
\end{proof}

Finally, it is easy to combine the two proof techniques and to establish
the following theorem for the case where $x$ is not exactly sparse and
the measurements are noisy.

\begin{theorem}\label{thm:non-sparse-noisy}
Suppose $x \in \R^n$ and that $\s_k(x,\nmm{\cdot}_1) \leq \d$.
Write $x = x_d + x_r$ where $x_d$ is the dominant part of $x$ consisting
of its $k$ largest components, and $x_r = x - x_d$ is the residual.
Let $y = Ax + \eta$ where $\nmm{\eta}_0 \leq M$,
and $A$ satisfies the main assumption with $q > 2k(r-1)+ 2M$.
Apply the modified recovery algorithm.
Then (i) $\supp(\xh) = \supp(x_d)$, and (ii) $\nmm{\xh - x_d}_\infty \leq \d$.
\end{theorem}

\subsection{Construction of a Binary Measurement Matrix}

The results presented until now show that the key to the procedure is
the construction of a binary matrix $A$ that satisfies the main assumption.
In this subsection, it is shown that
previous work by DeVore \cite{DeVore07} provides a simple
recipe for constructing a binary matrix with the desired properties.
Note that \cite{DeVore07} was the first paper to propose a completely
deterministic procedure for constructing a matrix that satisfies the
restricted isometry property.
It is shown in this section that DeVore's matrix is also a special
case of the bi-adjacency matrix of an expander graph.
Therefore the DeVore matrix acts as a bridge between two distinct
compressive sensing algorithms.
% Therefore the reinterpretation of this construction to achieve
% sparse recovery in a highly efficient manner is one of the contributions
% of the present paper.

We now describe the construction in \cite{DeVore07}.
Suppose $q$ is a prime number or a power of a prime number, and
let $\Fbb_q$ denote the finite field with $q$ elements.
% Then $\Fbb_q$ denotes the set $\{ 0 , 1 , \ldots , q-1 \}$
% with arithmetic modulo $q$, and is a field.
% If $q = p^s$ where $p$ is a prime number and $s > 1$,
% then it is possible to define
% a finite field $\Fbb_q$ with $q$ elements
% by identifying an irreducible polynomial $\psi$ of
% degree $s$ with coefficients in $\Fbb_p$, and then doing arithmetic on
% the polynomial ring $\Fbb_p[z]$ modulo $\psi(z)$.
% From an implementation standpoint, it is far more efficient to choose $q$
% to be a prime number, though in theory it is also possible to choose $q$
% to be a prime power.
% Suppose $q$ is a prime number or a power of a prime number, and
% let $\Fbb_q$ denote the finite field with $q$ elements.
% Note that for rapid implementation in \texttt{Matlab},
% we would choose $q$ to be a prime number,
% in which case $\Fbb_q$ denotes the set $\{ 0 , 1 , \ldots , q-1 \}$
% with arithmetic modulo $q$.
% However, the theory itself is applicable even to the case where $q$ is
% a power of a prime number but not a prime number.
% 
Suppose $a$ is a polynomial of degree $r-1$ or less with coefficients in
$\Fbb_q$, and define its ``graph'' as the set of all pairs $(x,a(x))$ as $x$
varies over $\Fbb_q$.
Now construct a vector $u_a \in \bi^{q^2 \times 1}$ by setting the entry
in row $(i,j)$ to $1$ if $j = a(i)$, and to zero otherwise.
To illustrate, suppose $q = 3$, so that $\Fbb_q = \{ 0 , 1 , 2 \}$
with arithmetic modulo $3$.
Let $r = 4$, and let $a(x) = 1 + 2x + x^2 + x^3$.
% Because $q$ is a prime number,
% the associated field $\Fbb_q$ is just $\{ 0, 1, 2 \}$ with arithmetic
% modulo 3.
With this choice, we have that $a(0) = 1$, $a(1) = 2$, and $a(2) = 2$.
The corresponding $9 \times 1$ column vector
has $1$'s in positions $(0,1),(1,2),(2,2)$ and zeros elsewhere.
This construction results in a $q^2 \times 1$ column vector $u_a$
that consists of $q$ blocks of size $q \times 1$, each of which
contains a single $1$ and $q-1$ zeros.
Therefore $u_a$ contains $q$ elements of $1$ and the rest equal to zero.
% In fact the column vector has even more structure, though we do not
% make use of it.

Now let $\Pi_{r-1}( \Fbb_q )$ denote the set of all polynomials of degree $r-1$
or less with coefficients in $\Fbb_q$.
In other words,
\bd
\Pi_{r-1} ( \Fbb_q ) := \left\{ a(x) = \sum_{i=0}^{r-1}
a_i x^i , a_i \in \Fbb_q
\right\} .
\ed
Note that $\Pi_{r-1} ( \Fbb_q )$ contains
precisely $q^r$ polynomials, because each of the $r$
coefficients can assume $q$ different values.\footnote{If the leading
coefficient of a polynomial is zero, then the degree would be less than $r$.}
Now define
\be\label{eq:23}
A := [ u_a , a \in \Pi_{r-1} ( \Fbb_q ) ] \in \bi^{q^2 \times q^r} .
\ee

The following theorem from \cite{DeVore07}
shows that the matrix $A$ constructed as above
satisfies the main assumption, and also the RIP with appropriately
chosen constants.

\begin{theorem}\label{thm:DeVore}
(See \cite[Theorem 3.1]{DeVore07})
For the matrix $A \in \bi^{q^2 \times q^r}$ defined in \eqref{eq:23},
we have that
\be\label{eq:24}
\IP{u_a}{u_b} \leq r-1
\ee
whenever $a,b$ are distinct polynomials in $\Pi_{r-1} ( \Fbb_q )$.
Consequently, if we define the column-normalized matrix
$A' = (1/\sqrt{q}) A$, then $A$
	satisfies the RIP of order $k$ with constant $\d_k \leq ((k-1)(r-1))/q$.
% Consequently, the column-normalized matrix $A' = (1/\sqrt{q}) A$
% satisfies the RIP of order $k$ with constant $\d_k = (k-1)(r-1)/p$.
\end{theorem}

In Theorem \ref{thm:rand-exp} it is shown that randomly generated
left-regular graphs are expanders, but this result is not particularly useful.
It is therefore of interest to have available methods that are
\textit{guaranteed} to generate expander graphs, even if the number of output
vertices is larger than with random constructions.
One such procedure is given in \cite{Guruswami-et-al09}.
We now describe this construction, and then show that the DeVore 
construction is a special case of it.
% Next it is shown that the DeVore construction is a special case of
% a method given in \cite{Guruswami-et-al09} for constructing expander graphs.
The construction in \cite{Guruswami-et-al09} is as follows:
Let $h \geq 2$ be any integer.
Then the map $\GA: \Fbb_q^r \times \Fbb_q \ap \Fbb^{s+1}$ is defined as
\be\label{eq:622}
\GA(f,y) := [ y, f(y) , f^h(y) , f^{h^2}(y) , \ldots , f^{h^{s-1}}(y) ] .
\ee
An alternate way to express the function $\GA$ is:
\bd
\GA(f,y) = [ y, ( f^{h^i}(y) , i = 0 , \ldots , s-1 ) ] .
\ed
In the definition of the function $\GA$,
$y$ ranges over $\Fbb_q$ as the ``counter,'' and the above graph
is left-regular with degree $q$.
The set of input vertices is $\Fbb_q^r$, consisting of polynomials in some
indeterminate $Y$ with coefficients in $\Fbb_q$ of degree no larger than $r-1$.
The set of input vertices has cardinality $q^r$.
The set of output vertices is $\Fbb^{s+1}$ and each output vertex is
an $(s+1)$-tuple consisting of elements from $\Fbb_q$.
The set of output vertices has cardinality $q^{s+1}$.
Note that the graph is $q$-left regular in that every input vertex has
exactly $q$ outgoing edges.

\begin{theorem}\label{thm:Guru-et-al}
(See \cite[Theorem 3.3]{Guruswami-et-al09}.)
For every pair of integers $h,s$, the bipartite graph defined in \eqref{eq:622}
is a $(h^s,1 - \beta)$-expander with
\be\label{eq:623}
\beta = \frac{ (r-1)(h-1)s }{q} 
\ee
whenever
\bd
h < \frac{q}{s(r-1)} + 1 .
\ed
\end{theorem}
Note that the inequality simply ensures that $\beta > 0$.
% \subsection{Relationship of DeVore's Construction to Expander Graphs}

Now we relate the construction of DeVore with that in \cite{Guruswami-et-al09}.

\begin{theorem}\label{thm:our}
The matrix $A$ constructed in \cite{DeVore07} is a special case of the graph in
Theorem \ref{thm:Guru-et-al} with $s = 1$, and any value for $h$.
Therefore a bipartite graph with the biadjacency matrix of \cite{DeVore07} 
is a $(h,1 - \beta)$-expander with
\be\label{eq:623a}
\beta = \frac{ (r-1)(h-1) }{q} 
\ee
whenever
\bd
h < \frac{q}{r-1} + 1 .
\ed
\end{theorem}

\begin{proof}
Suppose that $s = 1$ and that $h$ is any integer.
In this case each polynomial $f$ with coefficients in $\Fbb_q$
of degree $r-1$ or less gets mapped into the pair $(y,f(y))$ as $y$
ranges over $\Fbb_q$.
This is precisely what was called the ``graph'' of the polynomial $f$
in \cite{DeVore07}.
\end{proof}

\section{Computational Results}

Theorems \ref{thm:DeVore} and \ref{thm:our} show that the measurement
matrix construction proposed in \cite{DeVore07}
% proposed by us (based on DeVore's construction)
falls within the
ambit of both the restricted isometry property as well as expander graphs.
In other words, the binary DeVore's measurement matrix satisfies both RIP-2
and RIP-1 due to its construction and expander graph nature, respectively.
% Specifically, it follows from \eqref{eq:24} that
% the column-normalized matrix $A' = (1/\sqrt{q}) A$
% satisfies the RIP of order $k$ with constant $\d_k = (k-1)(r-1)/p$.
Hence this matrix can be used together with $\ell_1$-norm minimization,
the expander graph algorithm of Xu-Hassibi, as well as our
proposed algorithm.
In this section we compare the performance of all three algorithms
using the DeVore construction.
Note however that the number of rows of the matrix (or equivalently,
the number of measurements) will vary from one method to another.
This is discussed next.

% to achieve compressive sensing.
% Moreover, because the matrix $A$ satisfies the main assumption, it can
% also be used in our new algorithm.
% In this section, we compare the computational performance of our new
% algorithm with $\ell_1$-norm minimization as well as the Xu-Hassibi algorithm.

\subsection{Number of Measurements Required by Various Methods}

% Theorem \ref{thm:DeVore} shows that the measurement matrix
% used in the New Recovery Algorithm satisfies the RIP, while Theorem
% \ref{thm:our} shows that the same measurement matrix can also be justified
% on the basis of expander graphs.
In this subsection we compare the number of measurements required
by $\ell_1$-norm minimization, expander graphs, and our method.

In $\ell_1$-norm minimization, as shown in Theorem \ref{thm:DeVore},
the matrix $A$, after column normalization dividing each column by 
$\sqrt{q}$, satisfies the RIP with constant $\d_k = (k-1)(r-1)/q$.
Combined with Theorem \ref{thm:Cai-Zhang-1}, we conclude that
$\ell_1$-norm minimization with the DeVore construction achieves
robust $k$-sparse recovery whenever
\be\label{eq:41}
\frac{ ( \lceil tk \rceil -1 )(r-1) }{q} < \sqrt{ \frac{t-1}{t} } .
\ee
To maximize the value of $k$ for which the above inequality holds,
we set $r$ to its minimum permissible value, which is $r = 3$.
Also, we replace $ \lceil tk \rceil -1$ by its upper bound $tk$,
which leads to
\bd
\frac{2tk}{q} < \sqrt{ \frac{t-1}{t} } , \mbox{ or }
\frac{2k}{q} < \sqrt{ \frac { t-1 }{ t^3 } }.
\ed
Elementary calculus shows that the right side is maximized when
$t = 1.5$.
So the RIP constant of the measurement matrix must satisfy 
\bd
\d_{tk} < 
\sqrt{(t-1)/t} = 1/\sqrt{3} \approx 0.577 .
\ed
Let us choose a value of $0.5$ for $\d_{tk}$ to give some ``cushion.''
Substituting the values $t = 1.5, r = 3$ in \eqref{eq:41} and ignoring
the rounding operations finally leads to the condition
\be\label{eq:42}
\frac{3k}{q} < 0.5, \mbox{ or } q > 6k .
\ee

For expander graphs, we can calculate the expansion factor $1-\beta$
from Theorem \ref{thm:Guru-et-al}.
This gives
\bd
\beta = \frac{ (r-1) (h-1) s} { q } .
\ed
Since we wish the expansion factor $1 - \beta$
to be as close to one as possible,
or equivalently, $\beta$ to be as small as possible, we choose $s$
to be its minimum value, namely $s = 1$.
Now we substitute $r = 3$, $h = 2k$ (following \cite{Xu-Hassibi07}),
and set $1 - \e \geq 3/4$, or equivalently $\e \leq 1/4$.
This leads to
\bd
\frac{ 2 (2k-1) }{ q } \leq 1/4 , \mbox{ or } q \geq 8 (2k - 1) .
\ed
Finally, for the new algorithm, it has already been shown that
$q \geq 2(r-1)k = 4k$.
Note that, since the matrix $A$ has $q^3$ columns, we must also have that
$n \leq q^3$.

The required number of measurements for each of the
three algorithms are as shown in Table
\ref{table:meas}.
To facilitate the presentation, we introduce the notation
$\lceil x \rceil_p$ to denote the \textbf{smallest prime number}
that is no smaller than $x$.

\begin{table}
\bc
\btab{|c|c|c|}
\hline
Method & $q$ & $m$ \\
\hline
$\ell_1$-norm min. & $\lceil \max\{6k , n^{2/3} \} \rceil_p$ & $q^2$ \\
Expander Graph & $\lceil \max\{8(2k-1) , n^{2/3} \} \rceil_p$ & $q^2$ \\
New Algorithm & $\lceil \max\{4k , n^{2/3} \} \rceil_p$ & $q^2$ \\
\hline
\etab
\ec
\caption{Number of Measurements for Various Approaches}
\label{table:meas}
\end{table}

\subsection{Computational Complexity}

Because the new algorithm does not involve any iteration, it is very fast.
In this subsection we analyze the number of arithmetic operations involved
in implementing it.
For each index $j \in [n]$, there are in essence two steps: First,
to determine whether $j$ belongs to the support of the unknown vector,
and second, if $j$ does belong to $\supp(x)$, to determine the value of $x_j$.
This is achieved as follows:
For each index $j \in [n]$, the reduced vector $\yb_j$ is computed;
then $\yb$ is sorted in decreasing order of magnitude.
If $\yb_{(q+1)/2} = 0$, then $j \not \in \supp(x)$.
If $\yb_{(q+1)/2} \neq 0$, then the sorted vector is scanned over a window
of width $(q+1)/2$, and an index $a$ is chosen such that
$\yb_a = \yb_{a+(q-1)/2}$.
This is the value of $x_j$.
Thus, for each index $j$, the most time-consuming step is to sort $\yb$.
Since $\yb$ has $q$ components, the complexity is $O(q \log q)$, and since
$q = O(k)$, the complexity is $O(k \log k)$.
Since this has to be done $n$ times, the overall complexity is 
$O(nk \log k)$.
Note that the algorithm is fully parallelizable, in that each index $j$
can be processed separately and independently of the rest.

\subsection{Numerical Examples}

In this section we present a numerical example to compare the three methods.
We chose $n = 20,000$ to be the dimension of the unknown vector $x$.
Since all three methods produce a measurement matrix with $m = q^2$ rows,
we must have $q < 141 \approx \sqrt{20000}$,
because otherwise the number of measurements
would exceed the dimension of the vector!
Since the expander graph method requires the most measurements, the
sparsity count $k$ must satisfy $8(2k-1) < 141$, which gives $k \leq 9$.
However, if we try to recover $k$-sparse vectors with $k = 9$ using
the expander graph method, the number of measurements $m$ would be essentially
equal to the dimension of the vector $n$.
Hence we chose value of $k = 6$.
With this choice, the values of $q$ and the number of measurements are
shown in Table \ref{table:meas-2}.
Note that $q$ must be chosen as a prime number.

\begin{table}
\bc
\btab{|c|c|c|}
\hline
Method & $q$ & $m$ \\
\hline
$\ell_1$-norm min. & 37 & 1,369 \\
Expander graph & 89 & 7,921 \\
New algorithm & 29 & 841 \\
\hline
\etab
\ec
\caption{Number of measurements required for the numerical
examples with $n = 20,000$ and $k = 6$.}
\label{table:meas-2}
\end{table}

Having chosen the values of $n$ and $k$, we generated 100 different
$k$-sparse $n$-dimensional vectors, with both the support set of size $k$
and the nonzero values of $x$ generated at random.\footnote{Matlab codes
are available from the authors.}
As expected, both the expander graph method and the new algorithm recovered
the unknown vector $x$ \textit{exactly} in all 100 cases.
The $\ell_1$-norm minimization method recovers $x$ with very small error.
However, there was a substantial variation in the average time over the
100 runs.
Our algorithm took an average of 0.0951 seconds, or about 95 milliseconds,
$\ell_1$-norm minimization took 21.09 seconds, and the expander-graph
algorithm took 76.75 seconds.
Thus our algorithm was about 200 times faster than $\ell_1$-norm
minimization and about 800 times faster than the expander-graph algorithm.
% so the average error was computed.

As a final example, we introduced measurement noise into the output.
As per Theorem \ref{thm:noisy}, if $y = Ax + \eta$ where $\nmm{\eta}_0 
\leq M$, then it is still possible to recover $x$ exactly by increasing
the prime number $q$.
(Note that it is also possible to retain the same value of $q$ by
reducing the sparsity count $k$ so that $k+M$ is the same as before.)
Note that the only thing that matters here is the number of nonzero
components of the noise $\eta$, and not their magnitudes.
One would expect that, if the norm of the noise gets larger and larger,
our algorithm would continue to recover the unknown sparse vector exactly,
while $\ell_1$-norm minimization would not be able to.
In other words, our algorithm is tolerant to ``shot'' noise whereas
$\ell_1$-norm minimization is not.
The computational results bear this out.
We choose $n = 20,000$ and $k = 6$ as before, and $M = 6$, so that
we perturb the true measurement $Ax$ in six locations.
Specifically we chose $\eta = \al v$ where each component of $v$ is
normally distributed, and then increased the scale factor $\al$.
Each experiment was repeated with 100 randomly generated sparse vectors
and shot noise.
The results are shown in Table \ref{table:res-big}.

\begin{table}
\bc
\begin{tabular}{|c||c|c||c|c|}
\hline
\textbf{} & \multicolumn{2}{c|}{\textbf{New Algorithm}} & \multicolumn{2}{|c|}{\textbf{$\ell_1$-norm minimization}} \\ \hline
\textbf{Alpha}         & \textbf{Err.} & \textbf{Time} & \textbf{Err.} & \textbf{Time} \\ \hline
$10^{-5}$  & 0             & 0.1335                                            & 3.2887e-06    & 26.8822                                           \\ \hline
$10^{-4}$  & 0             & 0.1325                                             & 3.2975e-05    & 26.6398                                             \\ \hline
$10^{-3}$  & 0             & 0.1336                                            & 3.3641e-04    & 28.1876                                           \\ \hline
$10^{-2}$  & 0             & 0.1357                                        & 0.0033        & 23.1727                                           \\ \hline
$10^{-1}$ & 0             & 0.1571                                          & 0.033         & 28.9145                                          \\ \hline
10                    & 0             & 0.1409                                            & 1.3742        & 26.6362                                              \\ \hline
20                    & 0             & 0.1494                                        & 1.3967        & 26.5336                                         \\ \hline
\end{tabular}
\ec
\caption{Performance of new algorithm and $\ell_1$-norm minimization
with additive shot noise}
\label{table:res-big}
\end{table}

\section{Discussion and Conclusions}

In this paper we have presented a new algorithm for compressive sensing that
makes use of \textit{binary} measurement matrices and achieves
\textit{exact} recovery of sparse vectors, 
\textit{without any iterations}.
Exact recovery continues to hold even when the measurements are
corrupted by a noise vector with a sufficiently small support set;
this noise model is reminiscent of the model used in algebraic coding.
When the unknown vector is not exactly sparse, but is nearly sparse
with a sufficiently small residual, our algorithm exactly recovers
the support set of the dominant components, and finds an approximation
for the dominant part of the unknown vector.
Because our algorithm is non-iterative, it
executes orders of magnitude faster than
algorithms based on $\ell_1$-norm minimization 
and methods based on expander graphs (both of 
which require multiple iterations).
Moreover, our method requires a smaller number of measurements in comparison to
these two approaches when the measurement matrix is binary.
% that use binary measurement matrices.
On test examples of $k$-sparse $n$-dimensional vectors with
$k = 6$ and $n = 20,000$,
our algorithm executes roughly $1,000$ times faster than
the Xu-Hassibi algorithm \cite{Xu-Hassibi07} based on expander graphs,
and roughly $200$ times faster than $\ell_1$-norm minimization.

On the other hand, these two methods do have their own advantages over
the algorithm proposed here.
The Xu-Hassibi algorithm \cite{Xu-Hassibi07} and its extension in
\cite{JXHC09} can be used with \textit{any} expander graph with an
expansion factor that is sufficiently close to one.
In contrast, our algorithm makes use of \textit{a particular} family
of expander graphs whose bi-adjacency matrix satisfies the ``main
assumption.''
Similarly, if $\ell_1$-norm minimization is used to reconstruct
a vector, then a bound of the form \eqref{eq:14} holds no matter
what the unknown vector is.
In contrast, our error bounds require that the residual part of
the unknown vector must be sufficiently small compared to the
dominant part of the vector.
This might not be a serious drawback however, because the objective
of compressive sensing is to recover nearly sparse vectors, and not
arbitrary vectors.

\section*{Appendix}

In this appendix, we compare the number of measurements used by
probabilistic as well as deterministic methods to guarantee that
the corresponding measurement matrix $A$ satisfies the restricted 
isometry property (RIP), as stated in Theorem \ref{thm:Cai-Zhang-1}.
Note that the number of measurements is computed from the best available
sufficient condition.
In principle it is possible that matrices with fewer rows might also satisfy
the RIP.
But there would not be any theoretical justification for using such matrices.

In probabilistic methods, the number of measurements $m$ is
$O(k \log (n/k))$.
However, in reality the $O$ symbol hides a huge constant.
It is possible to replace the $O$ symbol by carefully collating the
relevant theorems in \cite{FR13}.
This leads to the following explicit bounds.

\begin{theorem}\label{thm:bound-prob}
Suppose $X$ is a random variable with zero mean, unit variance, and suppose
in addition that there exists a constant $c$ such that\footnote{Such
a random variable is said to be \textbf{sub-Gaussian}.
A normal random variable satisfies \eqref{eq:15b} with $c = 1/2$.}
\be\label{eq:15b}
E[ \exp(\th X)] \leq \exp(c \th^2) , \fa \th \in \R .
\ee
Define
\be\label{eq:15c}
\g = 2 , \zeta = 1/(4c) ,
\al = \g e^{- \zeta} + e^\zeta , \beta = \zeta ,
\ee
\be\label{eq:15d}
\ct := \frac { \beta^2 }{ 2 ( 2 \al + \beta ) } .
\ee
% in addition that $X$ satisfies \eqref{eq:554} for some constant $\ct$.
Suppose an integer $k$ and real numbers $\d , \xi \in (0,1)$ are specified,
and that $A = (1/\sqrt{m}) \Phi$, where $\Phi \in \R^{m \times n}$
consists of independent samples of $X$.
Then $A$ satisfies the RIP of order $k$ with constant $\d$ with probability
$\geq 1 - \xi$ provided
\be\label{eq:15e}
m \geq \frac{1}{ \ct \d^2 } \left( \frac{4}{3} k \ln \frac{en}{k}
+ \frac{14k}{3} + \frac{4}{3} \ln \frac{2}{\xi} \right) .
\ee
\end{theorem}

In \eqref{eq:15e}, the number of measurements $m$ is indeed $O(k \log (n/k))$.
However for realistic values of $n$ and $k$, the number of measurements
$n$ would be comparable to, or even to exceed,
$n$, which would render ``compressive'' sensing
meaningless.\footnote{In many papers on compressive sensing, especially those
using Gaussian measurement matrices, the number of measurements $m$ is
\textit{not} chosen in accordance with any theory, but simply picked
out of the air.}
For ``pure'' Gaussian variables it is possible to find improved
bounds for $m$ (see
Also, for binary random variables where $X$ equals $\pm 1$ with
equal probability, another set of bounds is available
\cite{Achlioptas03}.
While all of these bounds are $O(k \log (n/k))$,
in practical situations the bounds are not useful.

This suggests that it is worthwhile to study \textit{deterministic}
methods for generating measurement matrices that satisfy the RIP.
There are very few such methods.
Indeed, the authors are aware of only three methods.
The paper \cite{DeVore07} uses a finite field method to construct
a binary matrix, and this method is used in the present paper.
The paper \cite{Xu-Xu15} gives a procedure for choosing rows from a unitary
Fourier matrix such that the resulting matrix satisfies the RIP.
This method leads to the same values for the number of measurements $m$
as that in \cite{DeVore07}.
Constructing partial Fourier matrices is an important part of reconstructing
time-domain sparse signals from a limited number of frequency measurements
(or vice versa).
Therefore the results of \cite{Xu-Xu15} can be used in this situation.
In both of these methods, $m$ equals $q^2$ where $q$ is appropriately
chosen prime number.
Finally, in \cite{AHSC09} a method is given based on chirp matrices.
In this case $m$ equals a prime number $q$.
Note that the partial Fourier matrix and the chirp matrix are complex,
whereas the method in \cite{DeVore07} leads to a binary matrix.
In all three methods, $m = O(n^{1/2})$, which grows faster than
$O(k \log (n/k))$.
However, the constant under this $O$ symbol is quite small.
Therefore for realistic values of $k$ and $n$, the bounds for $m$ from
these methods are much smaller than those derived using probabilistic 
methods.

Table \ref{table:bounds-2} gives the values of $m$ for various
values of $n$ and $k$.
Also, while the chirp matrix has fewer measurements than the binary matrix,
$\ell_1$-norm minimization with the binary matrix runs much faster
than with the chirp matrix, due to the sparsity of the binary matrix.
In view of these numbers, in the present paper we used DeVore's 
construction as the benchmark for the recovery of sparse vectors.

\begin{table}
\bc
\btab{|rr|rrr|rr|}
\hline
$n$ & $k$ & $m_G$ & $m_{SG}$ & $m_A$ & $m_D$ & $m_C$ \\
\hline
  $10^4$   & 5  &  5,333  &    28,973  &  3,492  &  841 &    197 \\
    $10^4$   & 6  &  5,785  &    31,780  &  3,830  &  1,369 &    257 \\
      $10^4$   & 7  &  6,674  &    37,308  &  4,496  &  1,681 &    401 \\
        $10^4$   & 8  &  7,111  &    40,035  &  4,825  &  2,209 &    487 \\
  $10^4$   & 9  &  7,972  &    45,424  &  5,474  &  2,809 &    677 \\
    $10^4$   &  10  &  8,396  &    48,089  &  5,796  &  3,481 &    787 \\
    \hline
     $10^5$   &  10  & 10,025  &    57,260  &  6,901  &  3,481 &    787 \\
      $10^5$   &  12  & 11,620  &    66,988  &  8,073  &  5,041 &   1,163 \\
       $10^5$   &  14  & 13,190  &    76,582  &  9,229  &  6,889 &   1,601 \\
        $10^5$   &  16  & 14,739  &    86,061  & 10,372  &  9,409 &   2,129 \\
 $10^5$   &  18  & 16,268  &    95,441  & 11,502  &  11,449 &   2,707 \\
  $10^5$   &  20  & 17,781  &   104,733  & 12,622  &  16,129 &   3,371 \\
 \hline
 $10^6$   &   5  &   7,009  &    38,756  & 4,671  &  10,201 &   1,009 \\
 $10^6$   &  10  & 11,639  &    66,431  &  8,006  &  10,201 &   1,009 \\
 $10^6$   &  15  & 16,730  &    96,976  & 11,687  &  10,201 &   1,949 \\
 $10^6$   &  20  & 21,069  &   123,076  & 14,832  &  16,129 &   3,371 \\
 $10^6$   &  25  & 25,931  &   152,373  & 18,363  &  22,201 &   5,477 \\
 $10^6$   &  30  & 30,116  &   177,635  & 21,407  &  32,041 &   7,753 \\
 $10^6$   &  50  & 47,527  &   283,042  & 34,110  &  94,249 &  21,911 \\
 $10^6$   &  60  & 55,993  &   334,440  & 40,304  & 128,881 &  31,687 \\
 $10^6$   &  70  & 64,335  &   385,171  & 46,417  & 175,561 &  43,271 \\
 $10^6$   &  80  & 72,573  &   435,331  & 52,462  & 229,441 &  56,659 \\
 $10^6$   &  90  & 80,718  &   484,992  & 58,447  & 292,681 &  71,837 \\
 $10^6$   & 100  & 88,781  &   534,210  & 64,378  & 358,801 &  88,807 \\
 \hline
 \etab
 \ec
 \caption{Best available bounds for the
 number of measurements for various choices of $n$ and $k$
 using both probabilistic and deterministic constructions.
 For probabilistic constructions, the failure probability is $\xi = 10^{-9}$.
 $m_G, m_{SG}, m_A$ denote respectively the bounds on the number of measurements
 using a normal Gaussian, a sub-Gaussian with $c = 1/2$, and
 a bipolar random variable and the bound of Achlioptas.
 For deterministic methods $m_D$ denotes the number of measurements using
 DeVore's construction, while $m_C$ denotes the number of measurements using
 chirp matrices.}
 \label{table:bounds-2}
\end{table}

\bibliographystyle{IEEEtran}

\bibliography{Comp-Sens}

\end{document}